\RequirePackage{etex}
\documentclass{article}
\usepackage{graphicx} % Required for inserting images

\usepackage{amsmath,amsthm,amssymb,mathtools,booktabs}
\usepackage{enumitem}
\usepackage[linesnumbered,ruled,vlined,noend]{algorithm2e}
\usepackage{comment,todonotes}
\usepackage[numbers,sort]{natbib}
\usepackage{multirow}

\usepackage{caption}
\usepackage{subcaption}
\usepackage{authblk}
\usepackage{fullpage}

\usepackage{autonum}
\usepackage{enumitem}

\theoremstyle{definition}
\newtheorem{thm}{Theorem}[section]

\newtheorem{ex}{Example}[section]

\usetikzlibrary{patterns}

\title{Size-stochastic Knapsack Online Contention Resolution Schemes}
\author{Toru Yoshinaga}
\author{Yasushi Kawase}
\affil{The University of Tokyo}
\date{}

\begin{document}
\maketitle
\begin{abstract}
Online contention resolution schemes (OCRSs) are effective rounding techniques for online stochastic combinatorial optimization problems.
These schemes randomly and sequentially round a fractional solution to a relaxed problem that can be formulated in advance.
In this study, we propose OCRSs for online stochastic generalized assignment problems. 
In the problem of our OCRSs, sequentially arriving items are packed into a single knapsack, and their sizes are revealed only after insertion.
The goal of the problem is to maximize the acceptance probability, which is the smallest probability among the items being placed in the knapsack.
Since the item sizes are unknown beforehand, a capacity overflow may occur.
We consider two distinct settings: the hard constraint, where items that cause overflow are rejected, and the soft constraint setting, where such items are accepted.
Under the hard constraint setting, we present an algorithm with an acceptance probability of $1/3$ and prove that no algorithm can achieve an acceptance probability greater than $3/7$.
Under the soft constraint setting, we propose an algorithm with an acceptance probability of $1/2$ and demonstrate that this is best possible.
\end{abstract}
\section{Introduction}
The \emph{Online Stochastic Generalized Assignment Problem (OSGAP)} is a fundamental problem that represents several significant online problems, such as online knapsack problems~\cite{papastavrou1996dynamic,han2015randomized}, web advertising problems~\cite{MVV07,FKMMP09}, and online scheduling problems~\cite{DGV08,KRT97}.
In the OSGAP, multiple knapsacks with unit capacity are given first, and then a sequence of items with varying sizes and values arrives in an online manner.
Upon the arrival of an item, it must either be packed into a knapsack or discarded. The size of an item is revealed only after the placement in a knapsack. Instead, distributional information on the size of each item is available at the beginning.
The objective of the problem is to maximize the total value of the packed items. This problem can be considered as a scheduling problem where parallel machines process a sequence of jobs. Each knapsack corresponds to a machine with a maximum load, and each item corresponds to a job with processing time and profit. While the processing time for each job can be predicted, the actual processing time is determined only after the job is completed.

A promising approach for designing an online algorithm is to utilize an \emph{Online Contention Resolution Scheme (OCRS)}.
This scheme initially constructs a fractional solution to a relaxed problem from a priori information, and then stochastically and sequentially rounds it into a solution to the original stochastic online problem.
This concept was introduced by Feldman et al.~\cite{FSZ16}, 
and it has recently been applied to various problems, such as prophet inequality problems~\cite{LS18,JMZ22}, mechanism design of posted-price mechanisms~\cite{CM16,FSZ21}, and online matchings~\cite{EFGT20,FZHZ21,PRA22}.

Alaei et al.~\cite{AHL13} introduced an OCRS called the \emph{Generalized Magician Problem (GMP)}, for the OSGAP under the \emph{large-scale assumption}, in which the size of any arriving item does not exceed $1/k$ of the knapsack capacity, where $k$ is a prefixed positive integer.
The GMP is a stochastic online knapsack problem in which a unit-capacity knapsack is given initially, followed by a sequence of items with varying sizes arriving online. Upon an item's arrival, it must be packed into the knapsack or discarded. Here, items can be placed in the knapsack only when at least $1/k$ of the knapsack capacity remains. The goal of the problem is to maximize the \emph{acceptance probability}, which is the smallest probability among the items of being placed in the knapsack. If the acceptance probability for the GMP is $\gamma$, their OCRS guarantees a $\gamma$-fraction of the optimal value in expectation. Alaei et al.~\cite{AHL13} proved that the $\gamma$-conservative algorithm achieves the acceptance probability of $(1-1/\sqrt{k})$.
However, when $k=1$, no algorithm can guarantee a positive acceptance probability.

An OSGAP algorithm based on the GMP attempts to put an item in a knapsack only when there is no risk of overflow. However, in applications such as online scheduling, it is reasonable to try to put an item in a knapsack if there is a chance it will fit. Motivated by this, we consider a slightly different OCRS called the \emph{size-stochastic knapsack OCRS (KOCRS)} problem.
Unlike the GMP, the KOCRS problem assumes that each item can take any size, and we can try to put an item in a knapsack regardless of the remaining capacity.

\subsection{Our Results}
We provide a competitive analysis for the KOCRS problem.
Since a capacity overflow may occur in our setup, we consider constraint settings of \emph{hard} and \emph{soft}.
In the hard and soft constraint settings, items that cause overflow are rejected and accepted, respectively.
These settings can be interpreted in online scheduling as follows.
In the hard constraint setting, jobs that are not completed within a certain period are rejected and bring no profit.
On the other hand, in the soft constraint setting, jobs that start processing within the allotted time are accepted, regardless of whether they are completed by the time or not.

For the hard constraint setting, we provide an algorithm that has an acceptance probability of $1/3$ (Theorem~\ref{thm:hard_alg}).
We also devise an algorithm that has an acceptance probability of $3/7$ when the item sizes are limited to a binary distribution on $1$ and a sufficiently small positive real $\epsilon$.
Further, we prove that the acceptance probability is at most $3/7$ for any algorithm, even when the item sizes are limited to $\{\epsilon,1\}$ (Theorem~\ref{thm:hard_upper}).
For the soft constraint setting, we present an algorithm with the acceptance probability of $1/2$ (Theorem~\ref{thm:soft_alg}) and show that this is best possible even when item sizes are restricted to $\{\epsilon,1\}$ (Theorem~\ref{thm:soft_upper}).
Our results are summarized in Table~\ref{tab:OurResults}.

In Section~\ref{subsec:relation}, we discuss how algorithms for the KOCRS problem can be incorporated into the rounding framework for OSGAP, similar to those for GMP by Alaei et al.~\cite{AHL13}.
For the soft constraint setting, our OCRS implies an online algorithm for the OSGAP that gives a $1/2$-fraction of the optimal value of a relaxation problem in expectation.
For the hard constraint setting, our OCRS guarantees at least $1/3$-fraction of the optimal value if the values of the items are determined in advance.

\begin{table}[t]
    \centering
    \caption{Summary of the acceptance probability of the best algorithm for the KOCRS problem.}
    \label{tab:OurResults}
    \begin{tabular}{cccc}
    \toprule
    \multirow{2}{*}{Setting} & \multirow{2}{*}{Item size} & \multicolumn{2}{c}{Acceptance Probability}\\
     & & Lower bound & Upper bound\\
    \midrule
    Hard constraint & $[0,1]$ & $1/3$\ (Theorem~\ref{thm:hard_alg}) & $3/7$ (Theorem~\ref{thm:hard_upper})\\
    Hard constraint & $\{\epsilon,1\}$ & $3/7$\ (Theorem~\ref{thm:special_alg}) & $3/7$ (Theorem~\ref{thm:hard_upper})\\
    Soft constraint & $[0,1]$ & $1/2$\ (Theorem~\ref{thm:soft_alg}) & $1/2$ (Theorem~\ref{thm:soft_upper})\\
    Soft constraint & $\{\epsilon,1\}$ & $1/2$\ (Theorem~\ref{thm:soft_alg}) & $1/2$ (Theorem~\ref{thm:soft_upper})\\
    \bottomrule
    \end{tabular}
\end{table}

\subsection{Related Work}
%offline stochastic knaosack
% The (offline) stochastic knapsack problem has been considered in numerous models~\cite{BGK11,CSW93,D18,RT89}.
% In a typical setting, a policy maker is given a knapsack with a fixed capacity and a set of items  whose values and sizes follow probability distributions. The objective is to maximize the total value of the items placed in the knapsack while keeping the probability of overflow below a certain level.
The Generalized Assignment Problem (GAP), in the offline and deterministic problem setting, is APX-hard. The current best approximation factor for the problem is $1-1/e+\epsilon$ where $\epsilon\approx10^{-180}$~\cite{FV06}.

Marchetti-Spaccamela and Vercellis~\cite{MSV95} initiated the study of the online knapsack problem. In their work, they showed that the competitive ratio of any algorithm becomes arbitrarily small for the general case. Therefore, in the online knapsack problem, it is common to make natural restrictions based on real-world situations. For example, in the context of scheduling and web advertising, it is often assumed that the value and size of each item follow known probability distributions~\cite{A14,ZN08}.

The GMP and the KOCRS problem can be seen as generalizations of the classical prophet inequality problem~\cite{KS78} and the \emph{magician's problem (MP)}~\cite{A14}.
The MP is a special case of the GMP, in which the size of each item is either $1/k$ or $0$.
Alaei~\cite{A14} presented the $\gamma$-conservative algorithm, which attains the acceptance probability of $1-1/\sqrt{k+3}$ for the MP.
Note that, in the MP and the GMP, each item can be accepted only when there is no possibility of overflow, which differs from our setting.

%In terms of designing a knapsack OCRS for the online generalized assignment problem, Jiang et al.~\cite{JMZ22} took a similar approach to ours.
When constructing OCRSs for the online knapsack problem, a rounding technique that places each arriving item into the knapsack with a certain probability is often considered.
Jiang et al.~\cite{JMZ22} studied the knapsack OCRS problem, assuming that the item size follows a probability distribution, but the size of the item is revealed before its allocation.
Their best-fit algorithm places each item with probability $1/(3+e^{-2})\approx 0.319$.
Note that, in our setup, the size is determined only after insertion.
Our algorithms only guarantee the acceptance probability with respect to only items, rather than pairs of items and their size realizations.
Consequently, our setup is more challenging as the size realization is determined only after the insertion of each item, but easier in terms of the acceptance probability guarantee not being dependent on the realization of items.

Feldman et al.~\cite{FSZ21} provided a unified OCRS that can be applied to important families of constraints in online selection problems, such as matroids, matchings, and knapsacks.
For a size-deterministic online knapsack problem, their method guarantees a $3/2-\sqrt{2}(\approx 0.0858)$-fraction of the offline optimal value in expectation.
Their OCRS produces an extremely robust algorithm: their rounding technique guarantees a constant acceptance probability even against the almighty adversary---the strongest adversary model that generates inputs with complete knowledge of the algorithm, including the random behaviors.

Adamczyk and W{\l}odarczyk~\cite{AM18} pioneered the study of OCRSs in the random order model.
They provided a random-order contention resolution scheme for the intersection of matroids and knapsacks.
Our problem assumes an adversarial order of arrival for the items rather than a random order.

D\"{u}tting et al.~\cite{DFKL20} developed a general framework for stochastic online allocation problems.
Their threshold-based framework is applicable to a value-stochastic welfare maximization problem under knapsack constraints and derives an algorithm that guarantees a $1/5$-fraction of the expected optimal welfare.
Moreover, if all arriving items have sizes at most $1/2$ of the knapsack capacity, it guarantees a $1/3$-fraction of the expected optimal welfare.
In online knapsack problems, there have been several studies on algorithms that classify items by size and change the processing accordingly~\cite{FSZ21,SVX20}.
However, since the size of items is not determined upon arrival in our setting, such a technique is basically inapplicable.

\section{Preliminaries}
In this section, we formally define the KOCRS problem. %and the OSGAP.
%\subsection{The Size-Stochastic Knapsack OCRS Problem} 
%The KOCRS problem is defined as follows.
In the KOCRS problem, we are initially given a knapsack of unit capacity, and then $n$ items are given one by one.
The size $S_i$ of the $i$th item is a random variable drawn from a distribution $F_i$.
Upon the arrival of item $i$, the distribution $F_i$ is revealed, and we must immediately decide whether to try to insert the item into the knapsack or irrevocably discard it.
After attempting to insert the item into the knapsack, the size $S_i$ is determined.
If the item fits the knapsack, it is irrevocably placed in the knapsack.
It should be noted that we can try to insert an item into the knapsack even when there is a high probability of violating capacity.

When the total size of the inserted items exceeds the capacity, there are two settings where the last item is discarded or accepted.
We refer to such a last item as a \emph{critical item}.
The settings in which the critical item is discarded and accepted are called hard constraint and soft constraint settings, respectively.
In either setting, all items following the critical item are discarded.

We assume that $F_1,\dots,F_n$ are independent probability distributions on $[0,1]$.
In addition, it is promised that the total expected size of the items is less than or equal to the capacity of the knapsack, that is, $\sum_{i=1}^n\mathbb{E}[S_i]\le 1$.
The goal is to design a policy that accepts every item with probability at least $\gamma$, for a constant $\gamma\in[0,1]$ as large as possible.

The assumption that the total expected size of the items is at most $1$ does not provide a guarantee that all items can always be successfully inserted into the knapsack.
We confirm this fact with the following example, which also clarifies how the process differs between hard and soft constraint settings.
\begin{ex}\label{ex:procedure}
Suppose that there are three items: the first item has a size of $1/2$ with probability $1$, and the second and third items have sizes of $1$ with probabilities of $1/4$ and $0$ with probabilities of $3/4$, i.e.,
\begin{align}
S_1=1/2\quad\text{w.p. }1,\qquad
S_2=\begin{cases}1 & \text{w.p. }1/4,\\0 & \text{w.p. }3/4,\end{cases}\qquad
S_3=\begin{cases}1 & \text{w.p. }1,\\ 0 & \text{w.p. }3/4.\end{cases}
\end{align}
Note that the total expected size of the items is $1/2+1/4+1/4=1$.

Let us consider the strategy of greedily inserting incoming items into the knapsack.
Since the knapsack is initially empty, the first item is successfully inserted into the knapsack, and its size is revealed to be $1/2$ with probability $1$.

When attempting to insert the second item, it is revealed to be of size $1$ and exceeds the capacity of the knapsack with probability $1/4$.
Therefore, in the hard constraint setting, the second item is discarded with probability $1/4$. In the soft constraint setting, the second item is always accepted, but no further items can be accepted with probability $1/4$.

Finally, the procedure successfully inserts the third item only with probability $3/4\cdot 1/4=3/16$ and $3/4$ in the hard and soft constraint settings, respectively. 
Thus, it is not always possible to place all the items with probability $1$.
\end{ex}

We propose a policy that achieves a constant acceptance probability for any input sequence.
The policy determines the probability of inserting incoming items by tracking the probability distribution of the knapsack, which can be computed by using the size distributions of the items that have arrived.

It should be noted that an OCRS by Jiang et al.~\cite{JMZ22} addressed a setting in which the size of each item is revealed before deciding whether to insert the item into the knapsack. 
%They demonstrated that the $\gamma$-conservative algorithm, which selects a $\gamma$-measure of sample paths with low capacity utilization for each item, is not suitable.
They demonstrated that the $\gamma$-conservative algorithm, which inserts each item with probability $\gamma$ with priority on the least capacity utilized sample paths, is not suitable.
%Instead, they introduced the best-fit algorithm, which selects a $\gamma$-measure of sample paths with the highest capacity utilization for each item, on which the item still fits.
Instead, they introduced the best-fit algorithm, which prioritizes inserting each item with priority on the higher capacity utilized sample paths among those that can accommodate the item.
However, in our setting, we only know whether an item will fit in the knapsack or not after attempting to insert it. Therefore, the best-fit algorithm cannot be directly applied to our setting.

We will adopt a policy that takes into account the possibility of capacity overflows and selects the sample paths that utilize the most capacity while ensuring that the probability of each item being placed is exactly $\gamma$.

\section{Rounding Framework for the OSGAP}\label{subsec:relation}
We describe how the KOCRS problem can be employed to round a fractional solution to a relaxed problem of the OSGAP.
Specifically, we demonstrate that, by utilizing an algorithm for the KOCRS problem of acceptance probability $\gamma$, we can construct an algorithm for the OSGAP that achieves at least $\gamma$-fraction of the value of the fractional solution.
%Here, we define the OSGAP and describe the construction of a rounding framework for the OSGAP with a KOCRS algorithm.
For each integer $k$, let $[k]\coloneqq \{1,2,\dots,k\}$.

The OSGAP is defined as follows.
Initially, we are given $m$ knapsacks, each with a capacity of $1$. Then, $n$ items are presented in an online manner.
Each item has a pair of a value and a size distribution called a \textit{type}.
Item $i$ has $r_i$ possible types, and its actual type is initially known only probabilistically.
Upon the arrival of item $i\in[n]$, its type is determined to be $t\in[r_i]$ with probability $p_{it}$.
Then, the item must be inserted into one of the knapsacks or discarded immediately and irrevocably.
If the capacity constraint is violated in the process, the acceptance/rejection of the critical item is determined according to the setting: the critical item is accepted in the soft constraint setting and rejected in the hard constraint setting.
In both settings, items cannot be inserted into a knapsack after the critical item appears for it.
If item $i$ of type $t$ is inserted into knapsack $j$, its value is $v_{itj}$, and its size is determined according to distribution $F_{itj}$ on $[0,1]$.
We assume that all information on $r_i$, $p_{it}$, $v_{itj}$, and $F_{itj}$ (for each $i\in[n]$, $t\in[r_i]$, and $j\in [m]$) is known from the beginning.

Recall that Alaei et al.~\cite{AHL13} studied the OSGAP problem under an additional assumption called the large-capacity assumption: no item takes a size larger than $1/k$ of the knapsack capacity, where $k$ is a prefixed positive integer.
In contrast, our setting allows items to take any size.
In addition, Alaei et al.~\cite{AHL13} assumed that items are only placed in knapsacks with a remaining capacity of at least $1/k$, whereas we make no such assumption.

We construct an online algorithm for the problem by using a linear programming (LP) relaxation problem derived from a priori information.
The algorithm sequentially and stochastically rounds the optimal solution of the LP through the algorithms for the size-stochastic knapsack.
We analyze the expected instance of the problem where everything happens as per expectation. Specifically, for each item $i$ of type $t$ and knapsack $j$, we denote the expected size of the item as $\tilde{s}_{itj}=\mathbb{E}[S_{itj}]$, which is computable before the items arrive.
The linear programming is formalized as follows:
\begin{align}
\begin{array}{rll}
 \text{max.}&\sum_{i\in [n]}\sum_{t\in [r_i]}\sum_{j\in [m]} v_{itj}x_{itj}&\\
 \text{s.t.}&\sum_{i\in [n]}\sum_{t\in [r_i]} \tilde{s}_{itj}x_{itj}\leq 1 &\forall j \in[m],\\
&\sum_{j\in [m]} x_{itj}\leq p_{it} &\forall i\in[n],\,t\in[r_i],\\
&x_{itj}\in[0,1] &\forall i\in[n],\,t\in[r_i],\,j\in[m].
\end{array}\label{eq:LP}
\end{align}

We interpret the optimal solution $x_{itj}^*$ of the LP as the probability of inserting an item. Note that the expected total size of the item set sent to each knapsack is at most $1$. It would be ideal to insert items according to the relaxed problem's assignment, but it may cause a capacity violation, in a similar way to what we observed in Example~\ref{ex:procedure}. To avoid this, we use the KOCRS problem to realize the assignment with probability $\gamma\cdot x_{itj}^*$, and obtain $\gamma$-fraction in expectation. To be precise, we execute the following procedure: item $i$ of type $t$ is sent to knapsack $j$ with probability $x^*_{itj}/p_{it}$, and an algorithm for the KOCRS problem determines whether actually to place the item into the knapsack or discard it. We formally describe our procedure in Algorithm~\ref{alg:OCRS}.

\begin{algorithm}[thbp]
\KwIn{$\gamma\in[0,1],\ (F_{itj})_{i\in[n],t\in[r_i],j\in[m]}$, an algorithm for the KOCRS problem}
\KwOut{Item allocation in $m$ knapsacks and their size realization}
Solve LP~\eqref{eq:LP} and get the optimum $x^*$\;
For each knapsack $j\in[m]$, construct an instance $I_j$ of the KOCRS problem\;
\For{$i\gets1$ \KwTo $n$}{
  Let $t$ be the type of item $i$ and 
  select a knapsack (or nothing) $j^*\in [m]\cup\{\varnothing\}$ at random such that each knapsack $j\in[m]$ is chosen with probability $x_{itj}^*/p_{it}$\;
  \For{$j\gets1$ \KwTo $m$}{
    Let $F_{ij}$ be the CDF with $F_{ij}(s)=(1-\sum_{t'\in[r_i]}x^*_{it'j})+\sum_{t'\in[r_i]}x_{it'j}F_{it'j}(s)$\;
    Feed $F_{ij}$ to $I_j$ as the CDF of the $i$th item size\;
  }
  \uIf{Item $i$ is accepted in the algorithm for the KOCRS problem of $j^*$}{
    Assign item $i$ to knapsack $j^*$\; 
    Let the realization of the $i$th item be $0$ for $I_j$ with $j\ne j^*$ and $S_{itj^*}$ for $I_{j^*}$\;
  }
  \lElse{
    Discard item $i$
  }
}
\caption{Rounding Framework for the OSGAP with a KOCRS algorithm}\label{alg:OCRS}
\end{algorithm}

The combination of Algorithm~\ref{alg:OCRS} with an algorithm for the KOCRS problem of acceptance probability $\gamma$ has the following guarantee for the OSGAP.
\begin{thm}\label{thm:gamma_frac}
For the soft constraint setting, Algorithm~\ref{alg:OCRS} obtains at least $\gamma$-fraction of the optimal value of the LP if the acceptance probability of the utilized KOCRS problem is at least $\gamma$.
For the hard constraint setting, the same statement holds if the value of each item is independent of its type (i.e., $v_{itj}=v_{ij}$ for all $i,j,t$).
\end{thm}

\begin{proof}
The algorithm sent item $i$ of type $t$ to knapsack $j$ with probability $x^*_{itj}/p_{it}$
In the soft constraint setting, the item is successfully placed in the knapsack with probability $\gamma$.
Thus, the contribution of item $i$ to the objective value is $\sum_{t\in[r_i]}v_{itj}p_{it}(x^*_{itj}/p_{it})\cdot\gamma=\gamma\sum_{t\in[r_i]}v_{itj}x^*_{itj}$.

However, in the hard constraint setting, item $i$ of type $t$ sent to knapsack $j$ may be accepted with a probability less than $\gamma$ because the probability of overflow depends on the type.
Nevertheless, since the probability that item $i$ sent to knapsack $j$ fits is at least $\gamma$ in total, the contribution of item $i$ to the objective value is $v_{ij}p_{it}\sum_{t\in[r_i]}(x^*_{itj}/p_{it})\cdot\gamma=\gamma\sum_{t\in[r_i]}v_{ij}x^*_{itj}$ if the value of the item is independent of its type.

By summing up all the contributions of the items, we obtain $\gamma\sum_{i\in[n]}\sum_{j\in[m]}\sum_{t\in[r_i]}v_{itj}x^*_{itj}$, which is equal to the optimal value of the LP \eqref{eq:LP} multiplied by $\gamma$.
\end{proof}

\section{Analysis for the Hard Constraint Setting}\label{sec:hard}
This section analyzes the KOCRS problem with the hard constraint setting.
We first provide an instance of the hard constraint KOCRS problem, in which, for any algorithm, there exists an item with an acceptance probability at most $3/7~(\approx 0.429)$.
We then provide an algorithm that can accept each item with probability at least $1/3~(\approx 0.333)$. % by the aggressive policy.
Furthermore, we observe that each item can be accepted with a probability of at least $3/7$ if the sizes of the items are restricted to $1$ or a small positive $\epsilon\le 1/n$.

\subsection{Impossibility}
We demonstrate that no algorithm can accept every arriving item with a probability greater than $3/7$.
The fact implies that our problem belongs to an essentially different class from the MP with $k=1$ (i.e., the prophet inequality). Indeed, for the MP with $k=1$, we can ensure the acceptance probability of $1/2$.

\begin{thm}\label{thm:hard_upper}
    For any positive real $\delta$, there exists an instance of the hard constraint KOCRS problem where no algorithm has the acceptance probability of at least $3/7+\delta$.
    This impossibility holds even when item sizes are restricted to $1$ or a sufficiently small positive constant.
\end{thm}
\begin{proof}
    Let $\epsilon$ be a sufficiently small positive real.
    We will show that no algorithm has the acceptance probability of at least $3/7$ for the following instance with four items as $\epsilon$ approaches $0$:
    \begin{itemize}
        \item the size of the first item $S_1$ is $\epsilon$ with probability $1$,
        \item the sizes of the second and the third items $S_2$ and $S_3$ are $1$ with probability $1/2-\sqrt{\epsilon}$ and $\epsilon$ with probability $1/2+\sqrt{\epsilon}$, and
        \item the size of the last item $S_4$ is $\epsilon$ with probability $1$.
    \end{itemize}
% \begin{align}
% S_1=\epsilon\quad\text{w.p. }1,\qquad
% S_2=S_3=\begin{cases}1 & \text{w.p. }1/2-\sqrt{\epsilon},\\\epsilon & \text{w.p. }1/2+\sqrt{\epsilon},\end{cases}\qquad
% S_4=\epsilon\quad\text{w.p. }1.
% \end{align}
    Note that the total expected size of the items is
    \begin{align}
    \sum_{i=1}^4\mathbb{E}[S_i]
    &=\epsilon+2\cdot(1\cdot(1/2-\sqrt{\epsilon})+\epsilon\cdot(1/2+\sqrt{\epsilon}))+\epsilon\\
    &=1-2\cdot\epsilon^{1/2}+3\cdot\epsilon+2\cdot\epsilon^{3/2}
    < 1-2\cdot\epsilon^{1/2}+\frac{3}{2}\cdot\epsilon^{1/2}+\frac{2}{4}\cdot\epsilon^{1/2}
    = 1,
    \end{align}
    where the inequality holds if $\epsilon<1/4$.

    We will investigate the existence of an algorithm that accepts each item with a probability of at least $\gamma$. To this end, it suffices to consider the existence of an algorithm that accepts any item with exactly a probability of $\gamma$.
    %Figure~\ref{fig:thm31proof} depicts the procedure by which such an algorithm processes the considering instance.
    Since this algorithm accepts the first item with probability $\gamma$, the size of the item inside the knapsack after processing the first item becomes $\epsilon$ with probability $\gamma$, and $0$ with probability $1-\gamma$ (see Figure~\ref{subfig:hard-upper-a}).
    Suppose that the algorithm tries to insert the second item with probability $\alpha\in[0,1]$ under the condition that the first item is placed in the knapsack.
    If it tries to insert and the second item is revealed to be of size $1$, then the second item of size $1$ will be rejected since it violates the knapsack capacity.
    Hence, the probability that the first item is placed in the knapsack and the second item is also placed in the knapsack is $\gamma\cdot\alpha(1/2+\sqrt{\epsilon})$.
    To accept the second item with probability $\gamma$ overall, the algorithm inserts the second item with probability $(\gamma-\gamma\alpha(1/2+\sqrt{\epsilon}))/(1-\gamma)$, under the condition that the first item is discarded.
    The total size of the items in the knapsack, just after processing the second item, can be summarized as follows:
    the probability of exceeding $1$ (violating the capacity) is $\gamma\cdot \alpha/2+O(\sqrt{\epsilon})$,
    the probability of being exactly $1$ is $\gamma(1/2-\alpha/4)+O(\sqrt{\epsilon})$,
    the probability of being $\epsilon$ or $2\epsilon$ is $\gamma(3/2-3\alpha/4)+O(\sqrt{\epsilon})$,
    and the probability of being $0$ is $1-\gamma(2-\alpha/2)+O(\epsilon)$ (see Figures~\ref{subfig:hard-upper-b} and \ref{subfig:hard-upper-c}).
    For the third item, assume that the algorithm tries to insert it with probability $\beta\in[0,1]$, under the condition that the knapsack utilization is $\epsilon$ or $2\epsilon$.
    Similar to the insertion of the second item, the algorithm also tries to insert the third item into the knapsack with a utilization rate of $0$, as appropriate, to ensure that the third item can be inserted with an overall probability $\gamma$.
    The total size of the items in the knapsack, just after processing the third item, is in $\{\epsilon,2\epsilon,3\epsilon\}$ with probability $\gamma(2-3\alpha/4-3\beta/4)+O(\sqrt{\epsilon})$ and $0$ with probability $1-\gamma(3-\alpha/2-\beta/2)+O(\sqrt{\epsilon})$ (see Figure~\ref{subfig:hard-upper-d}).
    Here, the probability $1-\gamma(3-\alpha/2-\beta/2)+O(\sqrt{\epsilon})$ must be nonnegative.
    By considering $\epsilon\to 0$, we have 
    \begin{align}
    1-\gamma\left(3-\frac{\alpha}{2}-\frac{\beta}{2}\right)\ge 0. \label{eq:hard-upper1}
    \end{align}
    For the last item, it can be accepted with probability at most
    \begin{align}
    \MoveEqLeft
        \gamma\left(2-\frac{3\alpha}{4}-\frac{3\beta}{4}\right)+O(\sqrt{\epsilon})+1-\gamma\left(3-\frac{\alpha}{2}-\frac{\beta}{2}\right)+O(\sqrt{\epsilon})\\
        &=1-\gamma\left(1+\frac{\alpha}{4}+\frac{\beta}{4}\right)+O(\sqrt{\epsilon}).
    \end{align}
    To accept the last item with probability $\gamma$, 
    the value $1-\gamma\left(1+\frac{\alpha}{4}+\frac{\beta}{4}\right)+O(\sqrt{\epsilon})$ must be at least $\gamma$.
    Thus, by setting $\epsilon\to 0$, we have 
    \begin{align}
    1-\gamma\left(1+\frac{\alpha}{4}+\frac{\beta}{4}\right)\ge \gamma. \label{eq:hard-upper2}
    \end{align}

    \begin{figure}[htbp]
        \centering
        \begin{subfigure}[t]{0.48\textwidth}
        \centering
            \begin{tikzpicture}
                \draw[-stealth,thick](0,0)--(0,4)node[left]{1}--(0,4.5);
                \draw[-stealth,thick](0,0)--(4,0)node[below]{1}--(5.5,0)node[below right]{$W$};
                \draw[thick] (4,0) -- (4,0.1);
                \draw[thick] (0,4) -- (0.1,4);
                \draw (0,0) rectangle (1,2) node[pos=.5] {\scriptsize Item 1};
                \draw[thick] (1,0.1) -- (1,0) node[below] {$\epsilon$};
                \coordinate (a) at (1,0);
                \coordinate (b) at (1,2);
                \draw [thick] (a) to [bend right=30] node [fill=white, midway] { $\gamma$ } (b);
            \end{tikzpicture}
        \subcaption{The utilization of the knapsack after processing the first item. The horizontal axis represents the utilization of the knapsack, while the vertical axis represents the amount of probability. This figure represents that the knapsack utilization is $\epsilon$ with probability $\gamma$ and $0$ with probability $1-\gamma$.}\label{subfig:hard-upper-a}
        \end{subfigure}\hfill
        \begin{subfigure}[t]{0.48\textwidth}
        \centering
            \begin{tikzpicture}
                %座標
                \draw[-stealth,thick](0,0)--(0,4)node[left]{1}--(0,4.5);
                \draw[-stealth,thick](0,0)--(4,0)node[below]{1}--(5.5,0)node[below right]{$W$};
                %アイテム1
                \draw (0,0) rectangle (1,2) node[pos=.5] {\scriptsize Item 1};
                %アイテム2 alpha側
                \fill [lightgray] (1,0) rectangle (5,0.5);
                \draw (1,0) rectangle (5,0.5) node[pos=.5] {\scriptsize Item 2(1)};
                \draw (1,0.5) rectangle (2,1);
                \draw[dotted] (2,1) -- (5,1);
                \coordinate (c) at (5,0);
                \coordinate (d) at (5,1);
                \draw [thin] (c) to [bend right=60] node [fill=white, right=-0.2] { $\gamma\alpha$ } (d);
                \draw[thin](1.5,0.75)--(2.5,1.4)node[fill=white, right]{\scriptsize Item 2($\epsilon$)};
                %アイテム2 残り
                \draw (0,2.5) rectangle (1,3);
                \draw[dotted] (0,3) -- (5,3);
                \draw[dotted] (4,2) -- (5,2);
                \coordinate (e) at (5,2);
                \coordinate (f) at (5,3);
                \draw [thin] (e) to [bend right=60] node [fill=white, midway] { $\gamma-\gamma\alpha(\frac{1}{2}+\sqrt{\epsilon})$ } (f);
                \draw (0,2) rectangle (4,2.5)node[pos=.5] {\scriptsize Item 2(1)};
                \draw[thin](0.5,2.75)--(1.5,3.4)node[fill=white, right]{\scriptsize Item 2($\epsilon$)};
                \draw[thick] (4,0) -- (4,0.1);
                \draw[thick] (0,4) -- (0.1,4);
                \draw[thick] (1,0.1) -- (1,0) node[below] {$\epsilon$}; 
                \draw[dotted] (2,1) -- (2,0);
                \draw[thick] (2,0.1) -- (2,0) node[below] {$2\epsilon$};
                \draw[thick] (5,0.1) -- (5,0) node[below] {$1+\epsilon$};
            \end{tikzpicture}
        \subcaption{The utilization of the knapsack after processing the second item. For the second item, we denote it as Item 2(1) if its size realization is $1$ and Item 2($\epsilon$) if the realization is $\epsilon$. The item colored in gray is not accepted due to overflow.}\label{subfig:hard-upper-b}
        \end{subfigure}\\
        \begin{subfigure}[t]{0.48\textwidth}
        \centering
            \begin{tikzpicture}
                \draw[-stealth,thick](0,0)--(0,4)node[left]{1}--(0,4.5);
                \draw[-stealth,thick](0,0)--(4,0)node[below]{1}--(5.5,0)node[below right]{$W$};
                %W\ge1のsample paths
                \draw (0,0) rectangle (1,0.5) node[pos=.5] {\scriptsize Item 1};
                \fill [lightgray] (1,0) rectangle (5,0.5);
                \draw (1,0) rectangle (5,0.5) node[pos=.5] {\scriptsize Item 2(1)};
                \draw (0,0.5) rectangle (4,1)node[pos=.5] {\scriptsize Item 2(1)};
                %0<W<1のsample paths
                \draw (0,1) rectangle (1,2.5)node[pos=.5] {\scriptsize Item 1};
                \draw (0,2.5) rectangle (1,3);
                \draw (1,1) rectangle (2,1.5);
                \draw[thin](1.5,1.25)--(2,2.6)node[fill=white, right]{\scriptsize Item 2($\epsilon$)};
                \draw[thin](0.5,2.75)--(2,2.6);
                \draw[dotted] (1,3)--(3.5,3);
                \coordinate (g) at (3.5,3);
                \coordinate (h) at (3.5,1);
                \draw [thin] (g) to [bend left=60] node [fill=white, midway,right=-1] { $\gamma(\frac{3}{2}-\frac{3}{4}\alpha)+O(\sqrt{\epsilon})$ } (h);
                %W=0のsample paths
                \coordinate (i) at (0,3);
                \coordinate (j) at (0,4);
                \draw [thin] (i) to [bend right=60] node [fill=white, midway,right=-0.1] { $1-\gamma(2-\frac{1}{2}\alpha)+O(\sqrt{\epsilon})$ } (j);
                \draw[thick] (4,0) -- (4,0.1);
                \draw[thick] (0,4) -- (0.1,4);
                \draw[thick] (1,0.1) -- (1,0) node[below] {$\epsilon$}; 
                \draw[dotted] (2,1) -- (2,0);
                \draw[thick] (2,0.1) -- (2,0) node[below] {$2\epsilon$};
                \draw[thick] (5,0.1) -- (5,0) node[below] {$1+\epsilon$};
            \end{tikzpicture}
        \subcaption{The utilization of the knapsack after processing the second item, sorted by the knapsack utilization. The utilization of $\epsilon$ and $2\epsilon$ can be considered equivalent in the sense that the items that can be put in the future are the same.}\label{subfig:hard-upper-c}
%they can accept the third item only when its size realization is $\epsilon$.}
        \end{subfigure}\hfill
        \begin{subfigure}[t]{0.48\textwidth}
        \centering
            \begin{tikzpicture}
                %座標
                \draw[-stealth,thick](0,0)--(0,4)node[left]{1}--(0,4.5);
                \draw[-stealth,thick](0,0)--(4,0)node[below]{1}--(5.5,0)node[below right]{$W$};
                \path [pattern=north east lines, pattern color=black] (0,1.5)--(4,1.5)--(4,0.75)--(5,0.75)--(5,0)--(0,0);
                %W\ge1
                \draw (0,1.5)--(4,1.5)--(4,0.75) node[fill=white,left=1.3,fill opacity=0.9,draw opacity=1]{$W\ge1$}--(5,0.75)--(5,0)--(0,0);
                \coordinate (k) at (0,3);
                \coordinate (l) at (0,4);
                \draw [thin] (k) to [bend right=60] node [fill=white, midway,right=-0.1] { $1-\gamma(3-\frac{1}{2}\alpha-\frac{1}{2}\beta)+O(\sqrt{\epsilon})$ } (l);
                % 0<W<1
                \draw (0,3)--(1,3)--(1,2.5) node[below]{\footnotesize$\epsilon\le W\le3\epsilon$}--(2,2.5)--(2,2)--(3,2)--(3,1.5);
                \draw[dotted] (1,3)--(3,3);
                \coordinate (m) at (3,3);
                \coordinate (n) at (3,1.5);
                \draw [thin] (m) to [bend left=60] node [fill=white, pos=0.4,right=-0.5] { $\gamma(2-\frac{3}{4}\alpha-\frac{3}{4}\beta)+O(\sqrt{\epsilon})$ } (n);
                \draw[thick] (4,0) -- (4,0.1);
                \draw[thick] (0,4) -- (0.1,4);                
            \end{tikzpicture}
        \subcaption{The sorted utilization of the knapsack after processing the third item. 
        The last item cannot be added if the utilization is at least $1$ (hatched area).
        %In order to insert the last item with probability $\gamma$, at least $\gamma$ knapsacks with utilization rates in the range $0\le W\le3\epsilon$ (not shaded with diagonal lines) are required.
        }\label{subfig:hard-upper-d}
        \end{subfigure}
        \caption{The behavior of an algorithm for the instance that accepts every item with probability exactly $\gamma$ where $W$ is a random variable representing the utilization rate of the knapsack.}
        \label{fig:thm31proof}
    \end{figure}
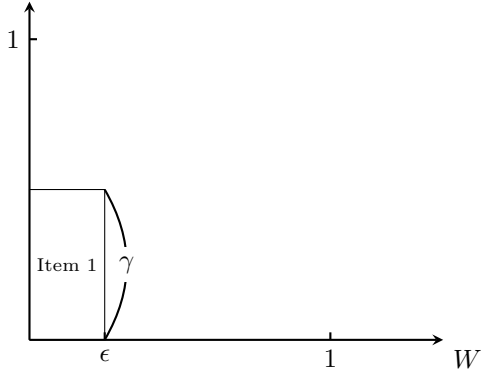
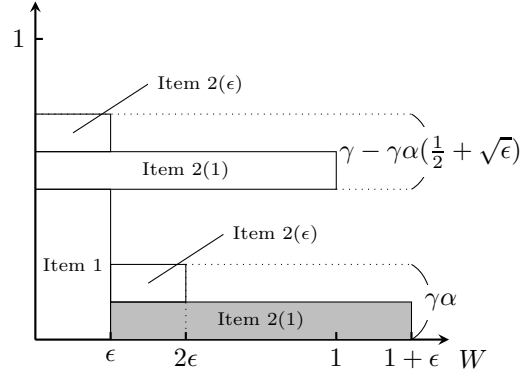
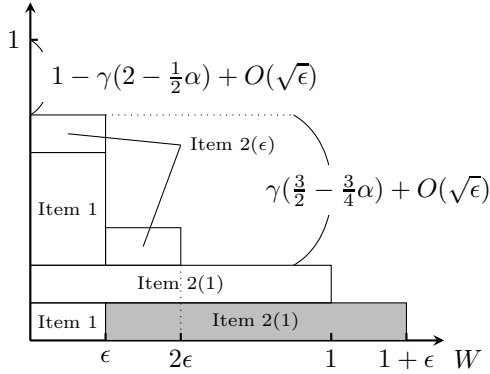
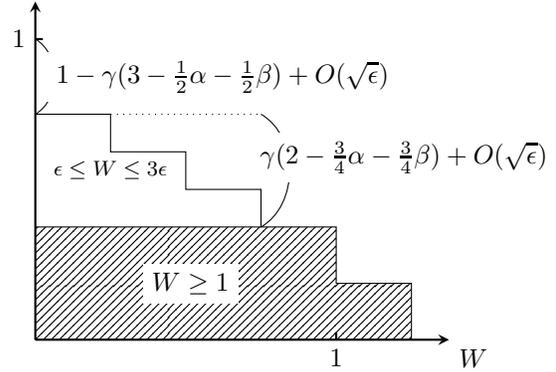
    By summing up inequality \eqref{eq:hard-upper2} twice and inequality \eqref{eq:hard-upper1}, we obtain $3-5\gamma\ge 2\gamma$.
    Hence, the upper bound of $\gamma$ is $3/7$.
    This implies that no algorithm can accept every item with a probability of at least $3/7+\delta$.
\end{proof}

\subsection{Algorithm}
Next, we prove that each item can be placed in the knapsack with a probability of at least $1/3$ by an aggressive policy.
The \emph{$\gamma$-aggressive algorithm} is an algorithm that selects sample paths so that items are accepted with probability $\gamma$, giving priority to sample paths that utilize more of the knapsack capacity.
Let $W_i$ be the random variable representing the total size of items in the knapsack just before processing the $i$th item, and let $F_{W_i}$ be the cumulative distribution function of $W_i$ (i.e., $F_{W_i}(w)=\Pr[W_i\le w]$). %knapsack utilization rates before $i$th item's arrival.
The algorithm uses a threshold $\theta_i\in[0,1]$ to determine whether to try to insert the $i$th item into the knapsack. If $W_i>\theta_i$, it tries to insert the item; if $W_i<\theta_i$, it discards the item; and if $W_i=\theta_i$, it tries to insert with an appropriate probability $q_i\in[0,1]$.
Note that, under a fixed $W_i=w$, the $i$th item fits the knapsack with probability $\Pr[S_i\le 1-w]$.
Thus, the algorithm successfully place the $i$th item with probability $\Pr[\theta_i<W_i\le 1-S_i]+q_i\cdot\Pr[\theta_i= W_i\le 1-S_i]$. 
%$\int_{\theta_i}^1F_i(1-w)\mathrm{d}F_{W_i}(w)+q_i\cdot\Pr[W_i=\theta_i]$.
To accept the $i$th item with probability $\gamma$, we set $\theta_i$ as the infimum of the set $\big\{\theta\in[0,1]\mid \Pr[\theta\le W_i\le 1-S_i]\ge\gamma\big\}$ 
and $q_i$ as $(\gamma-\Pr[\theta_i< W_i\le 1-S_i])/\Pr[\theta_i= W_i\le 1-S_i]$ (when $\Pr[\theta_i=W_i\le 1-S_i]=0$, we set $q_i$ to an arbitraly value in $[0,1]$).
%$q_i=\big(\gamma-\int_{\theta_i}^1F_i(1-w)\mathrm{d}F_{W_i}(w)\big)/\Pr[W_i=\theta_i]$ (when $\Pr[W_i=\theta_i]=0$, we can set $q_i$ to any value in $[0,1]$).
Here, such $\theta_i$ and $q_i$ exist if $\Pr[0<W_i\le 1-S_i]+\Pr[W_i=0]\ge \gamma$. 
%$\int_{0}^{1}F_{i}(1-w)\mathrm{d}F_{W_i}(w)+\Pr[W_i=0]\ge \gamma$.
Conversely, if such $\theta_i$ and $q_i$ exist, the algorithm successfully inserts the $i$th item with probability $\gamma$.
We will call the algorithm \emph{feasible} if such $\theta_i$ and $q_i$ exist for every item $i$.
For a given instance, if the $\gamma$-aggressive algorithm is feasible, it means that the algorithm successfully inserts every item with probability $\gamma$.

In summary, the $\gamma$-aggressive algorithm for the hard constraint setting tries to insert the $i$th item into the knapsack with the following probabilities:
\begin{itemize}
    \item if $\theta_i<W_i\le 1$, then the probability is $1$,
    \item if $W_i=\theta_i$, then the probability is $(\gamma-\Pr[\theta_i< W_i\le 1-S_i])/\Pr[\theta_i= W_i\le 1-S_i]$, and
    \item if $W_i<\theta_i$, then the probability is $0$,
\end{itemize}
where $\theta_i=\inf\big\{\theta\in[0,1]\mid \Pr[\theta\le W_i\le 1-S_i]\ge\gamma\big\}$.

We show that the $\gamma$-aggressive algorithm can guarantee each item to be accepted with a probability of $1/3$.
\begin{thm}\label{thm:hard_alg}
For any $\gamma\in (0,1/3]$, 
the $\gamma$-aggressive algorithm is feasible for any instance of the hard constraint KOCRS problem.
\end{thm}
\begin{proof}
    To analyze the $\gamma$-aggressive algorithm, we consider the following procedure:
    \begin{enumerate}
        \item Draw $X_1,\ldots, X_n$ and $X'_1,\ldots, X'_n$ independently each from the item size distributions $F_1,\ldots,F_n$.
        \item Set $T$ to $1$. If $\sum_{i=1}^nX_i\le 1$, return $T$.
        Otherwise, let $k$ be the smallest index such that $\sum_{i=1}^k X_i>1$.
        \item Increment $T$ by $1$. 
        If $X'_k+\sum_{i=k+1}^nX_i\le 1$, return $T$.
        Otherwise, let $\ell$ be the smallest index such that $X'_k+\sum_{i=k+1}^\ell X_i>1$.
        \item Set $k$ to $\ell$ and repeat step $3$.
    \end{enumerate}    
    We shall explain the correspondence between this procedure and our algorithm.
    The algorithm prioritizes the selection of sample paths in order of high knapsack utilization.
    %Hence, the measure of sample paths where the usage of the knapsack is in $(0,1]$ is always at most $\gamma$.
    Hence, the probability that the utilization of the knapsack is in $(0,1]$ is always at most $\gamma$.
    Consequently, the algorithm tries to insert an item if the usage of the knapsack is in $(0,1]$.
    For sample paths where an item is inserted, we may assume that the algorithm will try to insert the next item.
    The sample paths of $\gamma$ in which the first item is inserted correspond to step 2.

    If there are sample paths with capacity violations for an item, the algorithm selects sample paths of the corresponding proportion without item insertions.
    In the procedure, step 2 with $\sum_{i=1}^k X_i>1$ corresponds to such a sample path with a capacity violation at $k$th item (with item sizes $X_1,\dots,X_k$). 
    Step 3 corresponds to a newly selected sample path (with item sizes $X_k',X_{k+1},\dots,X_{\ell}$).
    The value $T$ represents the number of sample paths required to process the items when $(X_1,\dots,X_n,X'_1,\dots,X'_n)$ occurs.
    Thus, $\gamma\cdot\mathbb{E}[T]$ corresponds to the total amount of sample paths needed for the algorithm.
    It is sufficient to keep this value at most $1$ to guarantee the feasibility.

    We provide an upper bound of $\mathbb{E}[T]$.
    During the above procedure, the value $T$ is incremented by $1$ only when $\sum_{i=1}^kX_i>1$ in step 2 or $X_k'+\sum_{i=k+1}^\ell X_i>1$ in step 3 occurs. 
    Hence, for a fixed $(X_1,\dots,X_n,X'_1,\dots,X'_n)$, the number of times such an increment occurs is at most $\lfloor X_1+\cdots+X_n+X'_1+\cdots+X'_n\rfloor$.
    Thus, we obtain
    \begin{align}
    \mathbb{E}[T]
    &\le 1+\mathbb{E}\big[\lfloor X_1+\cdots+X_n+X'_1+\cdots+X'_n\rfloor\big]  \\
    &\le 1+\mathbb{E}[X_1+\cdots+X_n]+\mathbb{E}[X'_1+\cdots+X'_n]\le 3,
    \end{align}
    where the last inequality holds by the assumption that the total expected size of the items is at most $1$.
    Therefore, if we fix $\gamma\le1/3$, the algorithm is feasible for any instance since $\mathbb{E}[T]\cdot\gamma\leq 1$.
\end{proof}

Furthermore, we observe that our analysis in Theorem~\ref{thm:hard_alg} is best possible.
\begin{thm}
    For any $\gamma>1/3$, there exists an instance such that the $\gamma$-aggressive algorithm is not feasible.
\end{thm}
\begin{proof}
    Let $\delta$ be a sufficiently small positive real.
    We show this upper bound by the following instance, which consists of three items:
    \begin{itemize}
        \item the size of the first item $S_1$ is $\delta$ with probability $1$,
        \item the size of the second item $S_2$ is $1$ with probability $1-2\delta$ and $0$ with probability $2\delta$, and
        \item the size of the last item $S_3$ is $\delta$ with probability $1$.
    \end{itemize}
    Here, the total expected size of the items is
    $\mathbb{E}[S_1+S_2+S_3]=\delta+1\cdot(1-2\delta)+\delta=1$.

    The $\gamma$-aggressive algorithm inserts the first item into the knapsack with probability $\gamma$, thereby the usage of the knapsack becomes $\delta$ with probability $\gamma$.
    Since the algorithm tries to insert the second item if the first item is placed in the knapsack, a capacity violation occurs with probability $\gamma\cdot(1-2\delta)$ when processing the second item.
    Therefore, to ensure an overall insertion probability of $\gamma$, the algorithm selects $\gamma\cdot(1-2\delta)$ of the sample paths where the first item was rejected and then tries to insert the second item.
    %Thus, it selects $\gamma\cdot(1-2\delta)$-measure of sample paths such that the first item is not inserted and inserts the second item.
    The distribution of the usage of the knapsack at the end of processing the second item is given as follows: 
    \begin{enumerate}[label=(\roman*)]
    \item exceeding $1$ (capacity violation) with probability $\gamma(1-2\delta)$; 
    \item exactly $1$ with probability $\gamma(1-2\delta)^2$; 
    \item $\delta$ with probability $\gamma\cdot 2\delta$; \label{iii}
    \item $0$ with probability $1-\gamma-\gamma(1-2\delta)^2$. \label{iv}
    \end{enumerate}
    Since the last item can be successfully placed in the knapsack only under conditions \ref{iii} and \ref{iv}, the probability that the last item fits is at most $\gamma\cdot 2\delta+\big(1-\gamma-\gamma(1-2\delta)^2\big)$.

    By taking the limit as $\delta\to 0$, we have $1-2\gamma$ as the probability of accepting the last item.
    This probability is at most $1-2\gamma<1-2/3=1/3<\gamma$ by $\gamma>1/3$, which means that the algorithm is not feasible.
\end{proof}
From this theorem, we need to consider a different algorithm to achieve worst case performance better than $1/3$.

\subsection{Analysis for the Special Case}
In this subsection, we examine a special case of the hard constraint KOCRS problem where the size distributions are binary distributions of $1$ or $\epsilon~(\le 1/n)$.
For this setting, we present an algorithm that successfully puts every item with probability $3/7$.
Since the instance in the proof of Theorem~\ref{thm:hard_upper} belongs to this special case, the probability $3/7$ is best possible.

Our algorithm, \emph{$\gamma$-reserved}, keeps the probability that the usage of the knapsack is in $(0,1)$ at $\gamma$.
The algorithm first puts a virtual item of size deterministically $\epsilon$ into the knapsack with probability $\gamma$.
Let $W_i$ be the random variable denoting the usage of the knapsack just before processing the $i$th item, including the virtual item.
Also, let $p_i$ be the probability that the size of the $i$th item is $1$.
Then, the $\gamma$-reserved algorithm tries to insert the $i$th item into the knapsack with the following probabilities depending on the realization $W_i$:
\begin{itemize}
    \item if $W_i\geq 1$, then the probability is $0$ (in this case, we cannot put any more items),
    \item if $0<W_i<1$, then the probability is $\frac{1-p_i}{1-p_i+p^2_i}$, and
    \item if $W_i=0$, then the probability is $\frac{\gamma\cdot p_i}{1-p_i+p^2_i}\cdot\frac{1}{\Pr[W_i=0]}$.
\end{itemize}
Note that $\frac{p_i}{1-p_i+p_i^2}$ is in the interval $[0,1]$ whenever $p_i\in[0,1]$, as illustrated in Figure~\ref{fig:f(x)}.
We call the algorithm feasible if $\Pr[W_i=0]\geq\frac{\gamma\cdot p_i}{1-p_i+p^2_i}$.

\begin{figure}[htp]
\centering
\begin{tikzpicture}[xscale=4,yscale=4]
\draw[-stealth,thick](0,-0.1)--(0,1.2)node[left]{$y$};
\draw[-stealth,thick](-0.1,0)--(1.1,0)node[below]{$x$};
\draw(0,0)node[below left]{O};
\draw [domain=0:1,blue,thick] plot[smooth](\x, {(\x)/(1-\x+\x*\x)});
\node[right] at (1.1,0.9) {$f(x)=\frac{x}{1-x+x^2}$};
\draw(0.347,-0.1)node[]{$\alpha_0$};
\draw[dashed](0.347,0)--(0.347,0.448);
\fill[black] (0.347,0.448) circle (.2mm);
\draw[dashed](1,0)--(1,1);
\draw[dashed](0,1)--(1,1);
\draw(1,-0.1)node[]{$1$};
\draw(-0.1,1)node[]{$1$};
\end{tikzpicture}
\caption{The graph of $f(x)=\frac{x}{1-x+x^2}$ on $[0,1]$, where $(\alpha_0,f(\alpha_0))$ is an inflection point} \label{fig:f(x)}
\end{figure}
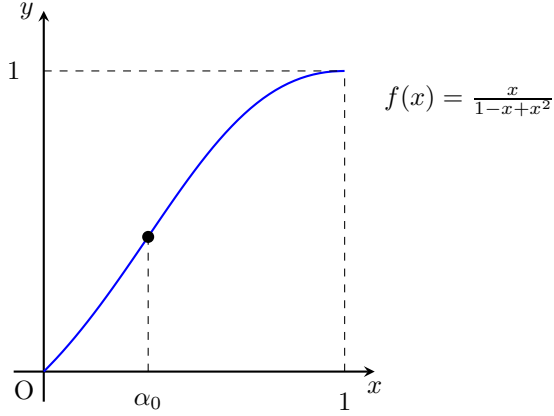

We show by induction that the algorithm keeps $\Pr[0<W_{i+1}<1]=\gamma$ if it is feasible.
The base case $\Pr[0<W_1<1]=\gamma$ clearly holds by the treatment of the virtual item.
For the inductive step, we get the probability $\Pr[0<W_{i+1}<1]$ as
\begin{align}
    \gamma\cdot\frac{1-p_i}{1-p_i+p_i^2}\cdot (1-p_i)
    +\gamma\cdot\left(1-\frac{1-p_i}{1-p_i+p_i^2}\right)
    +\frac{\gamma\cdot p_i}{1-p_i+p_i^2}\cdot (1-p_i)=\gamma.
\end{align}
In the left-hand side of the above equation, each term represents the probability that a sample path with $0 < W_i < 1$ is selected and the realized size is $\epsilon$, the probability that a sample path with $0 < W_i < 1$ is unselected at the $i$th step, and the probability that a sample path with $W_i=0$ is selected and the realized size is $\epsilon$.
Moreover, the algorithm successfully puts the $i$th item if it is feasible. In fact, the probability that the $i$th item is accepted is
\begin{align}
    \gamma\cdot\frac{1-p_i}{1-p_i+p_i^2}\cdot (1-p_i)
    +\frac{\gamma\cdot p_i}{1-p_i+p_i^2}=\gamma,
\end{align}
where the first term is the probability that a sample path with $0 < W_i < 1$ is selected and the realized size is $\epsilon$, and the second term is the probability that a sample path with $W_i=0$ is selected at the $i$th step, respectively.

We prove that the $\gamma$-reserved algorithm can guarantee each item to be accepted with a probability of $3/7$.
\begin{thm}\label{thm:special_alg}
For any $\gamma\in (0,3/7]$, 
the $\gamma$-reserved algorithm is feasible for any instance of the special case of the hard constraint KOCRS problem.
\end{thm}
\begin{proof}
    By the definition of the algorithm, we have $\Pr[W_1=0]=1-\gamma$ and $\Pr[W_{i+1}=0]=\Pr[W_i=0]-\frac{\gamma\cdot p_i}{1-p_i+p_i^2}$ for each $i$. Thus, to ensure feasibility, it is sufficient to satisfy that 
    \begin{align}
    \Pr[W_{n+1}=0]
    =1-\gamma-\sum_{i=1}^n\frac{\gamma\cdot p_i}{1-p_i+p_i^2}
    =1-\gamma\left(1+\sum_{i=1}^n\frac{p_i}{1-p_i+p_i^2}\right)
    \ge 0.
    \end{align}
    Thus, if the value of $\sum_{i=1}^n\frac{p_i}{1-p_i+p_i^2}$ is at most $4/3$, the feasibility is guaranteed for any $\gamma\in(0,3/7]$.

    Note that $\sum_{i=1}^np_i$ is at most $1$ since the total expected size of the items is at most $1$. Hence, our task is to demonstrate that the optimal value of the following optimization problem is at most $4/3$:
    % \begin{align}
    %     \text{max.}\quad \sum_{i=1}^n\frac{p_i}{1-p_i+p_i^2}\quad\text{s.t.}\quad \sum_{i=1}^np_i=1,\ p_i\ge 0~(i=1,2,\dots,n).\label{eq:special_opt}
    % \end{align}
\begin{align}
\begin{array}{rll}
 \text{max.}&\sum_{i=1}^n\frac{p_i}{1-p_i+p_i^2}&\\
 \text{s.t.}&\sum_{i=1}^np_i=1,&\\ 
            &p_i\ge 0          &(\forall i\in[n]).\label{eq:special_opt}
\end{array}\label{eq:special_opt}
\end{align}
    Fixing $n$, let $(p_1^*,\dots,p_n^*)$ be an optimal solution to the problem \eqref{eq:special_opt}. Without loss of generality, we may assume that $p_1^*\ge p_2^*\ge\dots\ge p_n^*$.
    We observe properties of the function $f(x)=\frac{x}{1-x+x^2}$ depicted in Figure~\ref{fig:f(x)}. This function has an inflection point at $x=\alpha_0\approx 0.3472$ and is concave in the range $[0,\alpha_0]$ and convex in the range $[\alpha_0,1]$.
    If $\alpha_0>p^*_i\ge p^*_j>0$ for some indices $i<j$, then the objective value can be increased by increasing $p^*_i$ by $\delta~(<\min\{\alpha_0-p^*_i,\,p^*_j\})$ and decreasing $p^*_j$ by $\delta$ because of the concavity of the function $f$ on $[0,\alpha_0]$.
    Hence, there is at most one index $i$ such that $p^*_i\in (0,\alpha_0)$.
    In addition, if $p^*_i>p^*_j\ge\alpha_0$ for some indices $i<j$, then the objective value can be increased by setting $p_i^*$ and $p_j^*$ to their average due to the convexity of the function $f$ on $[\alpha_0,1]$.
    Combining these facts with $\alpha_0>1/3$, the possible optimal solutions are limited to the following two cases:
    \begin{itemize}
    \item[(a)] $p^*_1=p^*_2=q$, $p^*_3=1-2q$, and $p^*_4=\dots=p^*_n=0$, where $q\in[\alpha_0,\,1/2]$,
    \item[(b)] $p^*_1=q$, $p^*_2=1-q$, and $p^*_3=\dots=p^*_n=0$, where $q\in[1-\alpha_0,\,1]$.
    \end{itemize}
    
    For case (a), the objective value is 
    \begin{align}
    \frac{2q}{1-q+q^2}+\frac{1-2q}{2q+(1-2q)^2}.
    \end{align}
    The derivative of the above with respect to $q$ equals
    \begin{align}
        \frac{2(1-q)(3q-1)(4q^4+q(2q-1)^2+(3q-1))}{(1-q+q^2)^2(1-2q+4q^2)^2},
    \end{align}
    which is positive when $q\in [\alpha_0,1/2] \subseteq [1/3,1]$.
    Hence, the objective value is maximized at $q=1/2$ over the interval $[\alpha_0,1/2]$, and the maximum value is $4/3$.

    For case (b), the objective value is 
    \begin{align}
    \frac{q}{1-q+q^2}+\frac{1-q}{q+(1-q)^2}=\frac{1}{(q-\frac{1}{2})^2+\frac{3}{4}}, 
    \end{align}
    which is maximized at $q=1-\alpha_0$, and the maximum value is smaller than $4/3$.
    Hence, the optimum value of \eqref{eq:special_opt} is $4/3$ and $\sum_{i=1}^n\frac{p_i}{1-p_i+p_i^2}$ is at most $4/3$.

    Therefore, setting $\gamma$ to $3/7$ makes the algorithm feasible for any instance.
\end{proof}

\section{Analysis for the Soft Constraint Setting}\label{sec:soft}
In this section, we analyze the soft constraint KOCRS problem.

\subsection{Impossibility}
We first observe that no algorithm can guarantee acceptance probability greater than $1/2$.
\begin{thm}\label{thm:soft_upper}
    For any positive real $\delta$, there exists an instance of the soft constraint KOCRS problem, where no algorithm has the acceptance probability of at least $1/2+\delta$.
    This impossibility holds even when item sizes are restricted to $1$ or a small positive constant $\epsilon$.
\end{thm}
\begin{proof}
Without loss of generality, we may assume that $\delta<1/2$.
Let $n=2+\lceil 1/\delta\rceil$ and $\epsilon$ be a positive real smaller than $\delta/n$.
Consider the following instance with $n$ items: %$n=1+\lceil 1/\delta\rceil$ items:
\begin{itemize}
\item the size of the first item $S_1$ is $1$ with probability $1-n\epsilon$ and $\epsilon$ with probability $n\epsilon$,
\item the size of the other item $S_i~(i=2,3,\dots,n)$ is $\epsilon$ with probability $1$.
\end{itemize}
Note that the total expected size of the items is 
\begin{align}
\sum_{i=1}^n\mathbb{E}[S_i]
=\left(1-n\epsilon+n\epsilon^2\right)+(n-1)\cdot \epsilon=1-\epsilon+n\epsilon^2\le 1,
\end{align}
since $\epsilon<\delta/n\le 1/n$.
We will prove the theorem by contradiction. 
Assume that there exists an online algorithm that accepts each item with probability at least $1/2+\delta$.
After processing the first item with this algorithm, the size in the knapsack is $1$ with a probability of at least 
\begin{align}
(1/2+\delta)\cdot(1-n\epsilon)>1/2+\delta-n\epsilon>1/2,
\end{align}
where the first and the second inequalities hold by $\delta<1/2$ and $\epsilon<\delta/n$, respectively.
Consequently, when processing each of the remaining items, a capacity overflow will occur with a probability of at least $\delta$.
Hence, after processing the $k$th item $(k=2,3,\dots,n)$, the probability of knapsack without violating capacity is at most $1-\delta\cdot (k-1)$.
Specifically, after processing all the items (i.e., $k=n$), the probability of knapsack without violating capacity is at most $1-\delta\cdot(\lceil 1/\delta\rceil+1)<0$. This is impossible, and therefore, no algorithm can accept every item with a probability of at least $1/2+\delta$ for this instance.
\end{proof}

\subsection{Algorithm}
Next, we prove that the aggressive policy can put every item into the knapsack with probability at least $1/2$, which is best possible by Theorem~\ref{thm:soft_upper}.
Recall that the $\gamma$-aggresive policy selects sample paths that utilize the most capacity so that the probability of each item being placed is exactly $\gamma$.
However, it should be noted that the $\gamma$-aggresive policy do not select sample paths that have a knapsack utilization rate equal to $1$ or more.
Let $W_i$ be the random variable representing the total size of items in the knapsack just before processing the $i$th item.
Then, the $\gamma$-aggresive algorithm for the soft constraint setting inserts the $i$th item into the knapsack with the following probabilities, according to the value of $W_i$:
\begin{itemize}
    \item if $W_i\ge 1$, then the probability is $0$ (in this case, we cannot put any more items in the knapsack),
    \item if $\theta_i<W_i< 1$, then the probability is $1$,
    \item if $W_i=\theta_i$, then the probability is $(\gamma-\Pr[\theta_i<W_i< 1])/\Pr[W_i=\theta_i]$,
    \item if $W_i<\theta_i$, then the probability is $0$,
\end{itemize}
where $\theta_i=\inf\big\{\theta\in[0,1]\mid \Pr[\theta\le W_i< 1]\ge\gamma\big\}$.
Note that, if such a $\theta_i$ exists (i.e., $\Pr[W_i< 1]\ge \gamma$), the algorithm successfully inserts the $i$th item into the knapsack with probability exactly $\gamma$.
We will call the algorithm \emph{feasible} when such a $\theta_i$ exists for all items $i$, i.e., $\Pr[W_i<1]\ge\gamma$ for all $i$.
For a given instance, if the $\gamma$-aggressive algorithm is feasible, it means that the algorithm successfully inserts every item with probability $\gamma$.
We show that the $1/2$-aggressive algorithm is feasible for any instance of the soft constraint KOCRS problem.

\begin{thm}\label{thm:soft_alg}
For any $\gamma\in (0,1/2]$, 
the $\gamma$-aggressive algorithm successfully inserts every item with probability $\gamma$ for any instance of the soft constraint KOCRS problem.
\end{thm}
\begin{proof}
We will prove this by mathematical induction.
%The base case holds since the first item is inserted with probability $\gamma$ by $\theta_1=0$ and $(\gamma-\Pr[\theta_1<W_1<1])/\Pr[W_1=\theta_1]=\gamma$.
Assume that each of the items up to the $(i-1)$st one has been successfully placed in the knapsack with probability $\gamma$.
The expected total size of items in the knapsack just before processing the $i$th item is $\mathbb{E}[W_i]=\sum_{j=1}^{i-1}\gamma\cdot\mathbb{E}[S_j]$.
By using the assumptions that $\sum_{j=1}^n\mathbb{E}[S_j]\leq 1$ and $\gamma\le 1/2$, we have $\mathbb{E}[W_i]\le 1/2$. Consequently, the probability that $W_i$ is less than one is at least
\begin{align}
    \Pr[W_i<1]=1-\Pr[W_i\ge 1]\ge 1-\mathbb{E}[W_i]\ge 1/2,
\end{align}
where the first inequality holds by Markov's inequality.
Therefore, the $\gamma$-aggressive algorithm is feasible, and it successfully inserts every item with probability $\gamma$.
\end{proof}

\bibliographystyle{abbrv}
\bibliography{main}
\end{document}